\documentclass[12pt]{article}

\usepackage[square,numbers]{natbib}
\usepackage{amsfonts}
\usepackage{graphicx}
\usepackage{enumitem}
\usepackage{centernot}
\usepackage{amsmath}
\usepackage{dsfont}  
\usepackage{amsthm}
\usepackage{color, soul}
\usepackage{tensor}

\usepackage{hyperref}
\hypersetup{
	colorlinks=true,
	linkcolor=blue,
	filecolor=black,      
	urlcolor=blue,
	citecolor=blue,
}

\newtheorem{Theo}{Theorem}[section]

\newtheorem{Proposition}[Theo]{Proposition}

\theoremstyle{definition} 
\theoremstyle{definition} 
\theoremstyle{definition} 

\newcommand{\be}{\begin{equation}}
	\newcommand{\ee}{\end{equation}}
\newcommand{\bea}{\begin{equation}\begin{aligned}}
		\newcommand{\eea}{\end{aligned}\end{equation}}
\newcommand{\x}{\mathbf{x}}

\newcommand{\Alg}{\mathcal{A}}

\title{Corrections to Wigner--Eckart Relations  by \\ Spontaneous Symmetry Breaking}
\author{
	{Carlo Heissenberg}$^{1,2}$ and {Franco Strocchi}$^3$\\
	{\footnotesize
		$^{1}$  Nordita, Stockholm University and  KTH Royal Institute of Technology,}\\ 
	{\footnotesize
		Roslagstullsbacken 23, 10691 Stockholm, Sweden;  carlo.heissenberg@su.se }\\
	{\footnotesize	$^{2}$  Department of Physics and Astronomy, 
		Uppsala University, 75108 Uppsala, Sweden
	}\\
{\footnotesize
		$^{3}$  Dipartimento di Fisica, Universit\`a di Pisa, Largo Pontecorvo 3, 56127 Pisa, Italy}	
}

\date{}

\begin{document}

	\maketitle
	
\begin{abstract}{The matrix elements of operators transforming as irreducible representations of an unbroken symmetry group $G$ are governed by the well-known Wigner--Eckart relations. In the case of infinitely-extended systems, with $G$ spontaneously broken, we prove that the corrections to such relations are provided by symmetry breaking Ward identities, and simply reduce to a tadpole term involving Goldstone bosons. The analysis extends to the case in which an explicit symmetry breaking term is present in the Hamiltonian, with the tadpole term now involving pseudo Goldstone bosons. An explicit example is discussed, illustrating the two cases. 
}
\end{abstract}

\noindent
{\textbf{Keywords}: spontaneous symmetry breaking; Wigner--Eckart; Ward identities}
{\textbf{PACS}: 02.20.-a; 03.65.Fd; 11.30.Qc}
{\textbf{Preprint}: NORDITA 2020-071}
	
\section{Introduction}
	After the fundamental pioneering book by E. P. Wigner \cite{Wigner}, the use of symmetries has become an important tool for the analysis of relevant properties of atomic systems and, more generally, of~quantum~systems.
	
	The existence of a group $G$ of symmetries for a given quantum system {a priori} provides structural information on its properties, \textit{e.g.}, (a) the classification of the states according to irreducible representations of $G$ ({multiplet structure}), (b) the {energy degeneracy} of the states of a given multiplet, (c) the existence of {selection rules} in the transitions induced by time evolution.
	
	Furthermore, the structure of the matrix elements of operators that transform as irreducible representations of $G$ is completely dictated by group-theoretical methods, explicitly by the so-called Wigner--Eckart relations.
	Such a strategy has been widely used, \textit{e.g.}, for the invariance under rotations in atomic physics and for isospin symmetry in nuclear physics. 
	
	In particular, the Wigner--Eckart relations may be used to compute {symmetry breaking effects} due to a small symmetry breaking term $H_1$ in the otherwise symmetric Hamiltonian, and therefore obtain the first-order corrections (energy splitting, mixing, etc.) to the symmetric picture. 
	
	This allows for a simple iterative procedure consisting of a perturbative expansion, in each step of which one meets the problem of computing matrix elements of powers of $H_1$, structurally governed by the corresponding Wigner--Eckart relations.
	
	In fact, this strategy is not confined to atomic systems, but it can also be exploited for infinitely extended systems in the case in which the symmetry is unbroken or softly broken by an explicit symmetry-breaking term.
	
	With the advent of spontaneous symmetry breaking, a mechanism typically appearing in infinitely extended systems, or in any case for systems admitting more than one inequivalent realization (see \textit{e.g.} \cite{FS_SSB,NoiSBTI}), the Wigner--Eckart strategy appears to be completely out of use; first because the symmetries of the Hamiltonian may be completely lost at the level of the states. Second, only in special cases the effect of spontaneous symmetry breaking may be considered small and treated by means of a perturbative expansion. In any case, for spontaneously broken symmetries the question arises as to how to obtain energy splittings and other corrections with respect to the Wigner--Eckart theorem.	
	
	The aim of this note is to show that, also in the case of a spontaneously broken continuous symmetry group, generalized Wigner--Eckart relations may be derived in a systematic and rigorous way, the exact corrections to the symmetric case being encoded in a ``tadpole'' term. 
	This may be regarded as an alternative {and complementary} strategy to the use of nonlinear Lagrangians, {where one instead describes the dynamics of the system by considering the most general effective field theory Lagrangian, compatible with the given unbroken symmetries, in which the broken symmetries are non-linearly realized} (see, \textit{e.g.}, \cite[Chap.~19]{WeinbergQFT2}). 
	
	Similar exact results are provided in the case in which the Hamiltonian consists of a symmetric part $H_s$ and a term $H_1$ that explicitly breaks the symmetry (ref.~\cite[Chap.~18]{FS_SSB}).  In particular, if one considers a symmetry breaking order parameter $A$ with $e^{i H_1 t} A e^{-i H_1 t} = e^{ih_A t}A$, then the corresponding tadpole term involves pseudo Goldstone bosons with energy spectrum satisfying $\omega(\mathbf k )\to h_A$ as $\mathbf k \to 0$. 
	
	A simple application is discussed, highlighting how, in the case with symmetric Hamiltonian, the tadpole term gives rise to the energy splitting, and therefore plays the role of a non-symmetric contribution. 
	With the addition of a symmetry breaking term, the example displays a general mechanism of a tadpole involving pseudo Goldstone bosons.

	\section{
		Correction to Wigner--Eckart Relations by SSB
	}\label{sec: tadpoles}
	As a guide for the solution of the problem raised above, we recall that, for quantum systems described by the standard Schr\"odinger representation, the Wigner--Eckart relations may be simply derived \cite{Sak} by using the Wigner theorem, according to which a continuous symmetry group $G$ is described by an $n$-parameter group of continuous unitary operators $U(\lambda)$, with $\lambda= (\lambda_1,\ldots, \lambda_n)$.
	In the following, we consider the case of compact $G$.
	By Stone's theorem, the corresponding infinitesimal generators $Q^a$ are well defined, as well as their action on the states and on the operators. In particular, for states and operators belonging to irreducible representations of $G$, one has
	\begin{equation}\label{generazione_ovvia}
		i Q^a \Psi_j = d\indices{^a_{jk}} \Psi_k\,,\qquad
		i Q^a \Phi_j = f\indices{^a_{jk}} \Phi_k\,,\qquad
		i [Q^a, A_m] = \mathcal{D}\indices{^a_{mn}} A_n\,,
	\end{equation}
	{ with $d\indices{^a_{jk}}$, $f\indices{^a_{jk}}$ and $\mathcal{D}\indices{^a_{mn}}$ the matrices associated with $Q^a$ in the corresponding representations.}
	Then,
	\begin{equation}\label{WE_base}
		\langle \Psi_j , i[Q^a, A_m] \Phi_k\rangle=
		-i\langle Q^a \Psi_j , A_m \Phi_k\rangle - i \langle \Psi_j, A_m Q^a \Phi_k\rangle
	\end{equation}
	leading to the Wigner--Eckart relation
	\begin{equation}
		\mathcal{D}\indices{^a_{mn}} \langle \Psi_j , A_n \Phi_k\rangle =
		- \bar d\indices{^{a}_{jl}} \langle \Psi_l , A_m \Phi_k \rangle -f\indices{^a_{kl}}\langle \Psi_j, A_m \Phi_l\rangle\,.
	\end{equation}
	In the case of spontaneous symmetry breaking, the possibility of describing the symmetry in terms of unitary operators { is precluded} and no infinitesimal generator $Q^a$ {exists}. Furthermore, the states do not  transform as irreducible representations of $G${, not even approximately}. 
	
	Thus, it seems quite difficult, if not impossible, to obtain suitably modified versions of the above relations. 
	A solution is provided by the algebraic approach , according to which a {spontaneously} broken symmetry is well defined on the relevant algebra $\Alg$ of canonical variables and/or observables, even if it fails to be implemented by unitary operators on the states. Hence, $G$ defines a group $\beta^{(\lambda)}$ of automorphisms of $\mathcal A$, { where $\lambda$ stands for the group parameters}, and we shall denote by $\delta^a$ the infinitesimal variation with respect to $\beta^{(\lambda)}$ 
	\be\label{limite1}
	\delta^a A
	\equiv \frac{d}{d\lambda^a}\beta^{(\lambda)}(A)\Big|_{\lambda=0},
	\ee
	for any $A\in\mathcal A$.
	Another important ingredient for the analysis of spontaneously broken continuous symmetries is their being locally generated {at the level of commutators, as we discuss in the next subsection}.
	
	\subsection{Spontaneous Symmetry Breaking with \\ Symmetric Hamiltonian}
	
	For infinitely extended systems, which typically exhibit spontaneous symmetry breaking, the above algebraic action of the symmetry may be exploited under the following conditions: one considers a representation of $\Alg$ given by a pure phase with a cyclic ground state vector $\Psi_\omega$, and a continuous symmetry group $G$ that commutes with time translations and is locally generated by conserved currents $j_\mu^a(\mathbf x, t)$. More precisely, for any $A\in\Alg$,
	\be\label{limite}
	\delta^a A = i \lim_{R\to\infty} [Q^a_R, A]\,,
	\ee
	where 
	\be\label{smear1}
	Q_R^a=\int_{|x|<R}d^dx\, j^a_0(\x,0)\,,
	\ee
	$d$ being the number of space dimensions. { The infrared cutoff is needed, in general, because the generator $Q^a$ may not exist as an operator in the Hilbert space of the system, as discussed below.}
	Actually, a smoother smearing may be necessary for the definition of $Q^a_R$, { which furthermore may not admit a sharp-time restriction as a distribution}, typically
	\be\label{smear2}
	Q^a_R=\int d^dx\, f_R(\x) \alpha(t) j^a_0(\x,t)\,,
	\ee
	where the test functions $f_R$ and $\alpha$ are infinitely differentiable and of compact support, and satisfy
	\be
	f_R(\mathbf x) \equiv f\bigg(\frac{|\mathbf x|}{R}\bigg),\quad 
	f(x) = \begin{cases}
		1 &\text{if }x<1\\
		0 &\text{if }x>1+\epsilon
	\end{cases}
	\quad
	\text{and}
	\quad
	\int \!\alpha(t)\,dt=1\,.
	\ee 
	{Such a smearing gives rise to well defined commutators, which, thanks to current conservation and locality, converge in the limit $R\to\infty$ even if the operators $Q_R^a$ do not.}
	{More generally, the} existence of the limit \eqref{limite}, with $Q^a_R$ given by \eqref{smear1} or more conveniently \eqref{smear2}, is guaranteed if the dynamics does not induce a long-range delocalization. This condition is automatically satisfied if the algebra $\Alg$ is generated by operators satisfying canonical (anti-)commutation relations, which are strictly local, and the delocalization induced by time evolution decays faster than $|\x|^{1-d}$ at large distances \cite{swiecaRANGE, FS_SSB}.
	
	Let $A_m$, $B_j$ and $C_k$ be elements of $\Alg$ transforming as irreducible representations of $G$, so that
	\be
	\delta^a B_j = d\indices{^a_{jk}} B_k\,,\qquad
	\delta^a C_j = f\indices{^a_{jk}} C_k\,,\qquad
	\delta^a A_m = \mathcal{D}\indices{^a_{mn}} A_n\,,
	\ee
	then, the algebraic identity
	\bea
	\langle B_j^\ast \Psi_\omega, \delta^a A_m C_k \Psi_\omega \rangle
	=\ &
	-\langle \delta^a B_j^\ast \Psi_\omega, A_m C_k \Psi_\omega\rangle
	-\langle B_j^\ast \Psi_\omega, A_m \delta^a C_k \Psi_\omega\rangle\\
	&+\langle \Psi_\omega,\delta^a(B_j A_m C_k )\Psi_\omega\rangle 
	\eea
	leads to
	\begin{equation}\label{sb_tad}
		\begin{aligned}
			\mathcal{D}\indices{^a_{mn}} \langle B_j^\ast \Psi_\omega , A_n C_k\Psi_\omega\rangle =&
			- \bar d\indices{^{a}_{jl}} \langle B_l^\ast \Psi_\omega , A_m C_k \Psi_\omega \rangle -f\indices{^a_{kl}}\langle B_j^\ast \Psi_\omega, A_m C_l\Psi_\omega\rangle\\
			&+\langle \Psi_\omega,\delta^a(B_j A_m C_k )\Psi_\omega\rangle\,.
		\end{aligned}
	\end{equation}
	If $G$ is unbroken, then the states $\Psi_j=B_j \Psi_\omega$ and $\Phi_k=C_k \Psi_\omega$ transform as the same irreducible representations of $B_j$ and $C_k$ respectively, thus reproducing the corresponding multiplet structure, and \eqref{sb_tad} coincides with the Wigner--Eckart relations \eqref{WE_base}. 
	
	The relevant point is that the corrections due to spontaneous symmetry breaking are fully encoded in the last term on the right-hand side of \eqref{sb_tad}, which is nontrivial whenever the ground-state expectations of $B_j A_m C_k$ are non-symmetric under the $a$-th subgroup of $G$.

	\section{Generalized Wigner--Eckart Relations and Ward Identities}
	
	Thanks to very important results by Ferrari and Picasso \cite{FP}, in the case of relativistic quantum field theories, the symmetry breaking term appearing on the right-hand side of \eqref{sb_tad} can be written in a very useful form, namely as a {tadpole} contribution.
	
	In the following, we discuss a version without using Lorentz symmetry.
	We consider a pure phase defined by a ground state vector $\Psi_\omega$, invariant under the Euclidean group $\mathbb E$ and $CPT$ symmetry.
	
	For simplicity, we consider the improper states $|\mathbf k, k_0\rangle$ of definite energy $k_0$ and momentum $\mathbf k$ (it is only cumbersome to reproduce the same conclusions in terms of wave packets), with the normalization $\langle\mathbf k, k_0 | \mathbf k', k'_0\rangle = \delta(\mathbf k - \mathbf k')$. We denote by $|A\rangle$ the state vector $A\Psi_\omega$.
	
	The spontaneous breaking of the $a$-th subgroup of $G$ implies the existence of Goldstone modes $|\mathbf k , k_0, a, \gamma\rangle$, \textit{i.e.}, gapless modes with infinite lifetime in the $\mathbf k \to 0$ limit, with $\gamma$ denoting possible additional quantum numbers. For simplicity, $\gamma$ will not be spelled out in the following (omitting the corresponding sum over $\gamma$ in the derived equations).
	
	\begin{Proposition}
		Under the above assumptions,
		let $\beta^{(\lambda)}$ be a group of symmetries generated by conserved currents $j_\mu^a(t, \mathbf x)$, which transform covariantly under $\mathbb E$ and $CPT$, and such that
		\begin{equation}\label{integrabil}
			\langle [j_\mu^a(t,\mathbf x), A]\rangle
			=
			\langle \Psi_\omega,[j_\mu^a(t,\mathbf x), A] \Psi_\omega\rangle
		\end{equation}
		are absolutely integrable in $\mathbf x$ as tempered distributions in $t$.
		Then (omitting the sum over $\gamma$),
		\be\label{dadim1}
		\langle \delta^a A \rangle = (2\pi)^{3/2}\lim_{\mathbf k\to0} F_a(\mathbf k^2, k_0)\, (\langle \mathbf k, k_0, a|A\rangle+\langle A | -\mathbf k, k_0, \bar a \rangle)\,,
		\ee
		where $F_a(\mathbf k^2, k_0)$ is defined by 
		\bea\label{daprov_2}
		-i (2\pi)^{-3/2} F_a(\mathbf k^2, k_0)=\langle j^a_0(0)|\mathbf k, k_0, a \rangle\,,
		\eea
		and $|\mathbf k , k_0, \bar a\rangle$ is given in terms of CPT transformation 
		\be
		|\mathbf k, k_0, \bar a \rangle = \mathrm{CPT} |\mathbf k, k_0, a\rangle\,.
		\ee
		Hence, Equation \eqref{sb_tad} takes the form
		\bea\label{tadpole_FINALE}
		&\langle B_i^\ast \Psi_\omega, \delta^a A_m C_j \Psi_\omega \rangle
		+\langle \delta^a B_i^\ast \Psi_\omega, A_m C_j \Psi_\omega\rangle
		+\langle B_i^\ast \Psi_\omega, A_m \delta^a C_j \Psi_\omega\rangle\\
		&=
		(2\pi)^{3/2}\lim_{\mathbf k \to 0} F_a(\mathbf k^2, k_0)\,  (\langle \mathbf k , k_0, a |B_i A_m C_j \rangle + \langle B_i A_m C_j |-\mathbf k, k_0, \bar a\rangle)
		\,,
		\eea
		so that the corrections to the Wigner--Eckart relations are given by the tadpole term on the right-hand side.	
	\end{Proposition}	
	\begin{proof}
		By the Goldstone theorem, one has
		\bea\label{interm}
		\langle \delta^a A \rangle =i(2\pi)^3 \lim_{\mathbf k \to0} 
		\big[&\langle j^a_0(0) |\mathbf k,k_0, a\rangle \langle \mathbf k,k_0, a|A\rangle\\
		&-
		\langle A |-\mathbf k,k_0,\bar a\rangle \langle -\mathbf k,k_0,\bar a| j^a_0(0)\rangle
		\big] 
		\,,
		\eea
		Invariance of $\omega$ under rotations implies that 
		\be
		\langle j^a_0(0)|\mathbf k, k_0, a \rangle = -i (2\pi)^{-3/2} F_a(\mathbf k^2, k_0)
		\ee
		for some function $F_a(\mathbf k^2, k_0)$
		and, by CPT invariance, 
		\be
		\langle j_0^a(0)|\mathbf k, k_0, a \rangle=
		-\langle \mathbf k, k_0, \bar a|j_0^a(0)\rangle\,.
		\ee 
		Therefore
		\bea\label{dariscr}	
		\langle \delta^a A \rangle
		= (2\pi)^{3/2} \lim_{\mathbf k \to 0} F_a(\mathbf k^2, k_0) \Big(\langle \mathbf k, k_0, a|A\rangle+\langle A |-\mathbf k, k_0, \bar a\rangle \Big)
		\,.
		\eea
	\end{proof}
	
	\paragraph{\textbf{Remarks:}}
	\noindent\textbf{1.} If the order parameter $A$ is CPT- and rotation-invariant,  $\langle A |-\mathbf k, k_0, \bar a\rangle=\langle A |\mathbf k, k_0, a\rangle$, Equation~\eqref{dariscr} reduces to 
	\be
	\langle \delta^a A \rangle = (2\pi)^{3/2}\lim_{\mathbf k\to0} F_a(\mathbf k^2, k_0)\, (\langle \mathbf k, k_0, a|A\rangle+\langle A |\mathbf k, k_0, a \rangle)\,,
	\ee
	\noindent\textbf{2.} For the special case of relativistic systems, $F_a(\mathbf k^2, k_0)=F_a(k_0^2-\mathbf k^2)=F_a(0)$, by the requirement of Lorentz invariance
	(see \cite{FP}).
	
	\noindent\textbf{3.} Invariance of the ground state under rotations implies that there exists a function $G_a(\mathbf k^2, k_0)$ such that
	\be
	\langle j^a_i(0) |\mathbf k, a\rangle = -i (2\pi)^{-3/2} k_i G_a(\mathbf k^2, k_0). 
	\ee
	and, by current conservation,
	\be
	|\mathbf k|^2 G_a(\mathbf k^2, k_0)= k_0 F_a(\mathbf k^2, k_0)\,.
	\ee
	The spatial components of the charge commutator then satisfy
	\bea
	&{i}\big(\langle j^a_i(0) |\mathbf k, k_0, a\rangle \langle  
	\mathbf k, k_0, a|A\rangle 
	-
	\langle A |-\mathbf k, k_0, \bar a\rangle \langle  
	-\mathbf k, k_0, \bar a| j^a_i(0) \rangle \big)
	\\
	&= (2\pi)^{-3/2}\frac{ k_i k_0}{|\mathbf k|^2} F_a(\mathbf k^2, k_0) \Big(\langle \mathbf k, k_0, a|A\rangle-\langle A |-\mathbf k, k_0, \bar a\rangle \Big)
	\eea
	and, by continuity as $\mathbf k\to0$, which follows by the charge integrability condition \eqref{integrabil}, we conclude
	\be
	\lim_{\mathbf k\to0}\frac{k_0}{|\mathbf k|}F_a(\mathbf k^2, k_0) \Big(\langle \mathbf k, k_0, a|A\rangle-\langle A |-\mathbf k, k_0, \bar a\rangle \Big)=0\,.
	\ee
	If the energy-momentum dispersion relation of the Goldstone states is such that $k_0/|\mathbf k|$ does not tend to $0$ as $\mathbf k \to 0$,~then 
	\be
	F_a(\mathbf k^2,k_0) (\langle \mathbf k, k_0, a|A\rangle-\langle A |-\mathbf k, k_0, \bar a\rangle)
	\ee
	must go to zero as $\mathbf k \to 0$ (as it happens in particular for Goldstone bosons of type I, see \cite{nielsenCOUNTING}, or in the relativistic case) and Equation \eqref{dadim1} takes the simpler form
	\be
	\langle \delta^a A \rangle = 2(2\pi)^{3/2}\lim_{\mathbf k\to0} F_a(\mathbf k^2, k_0)\langle \mathbf k, k_0, a|A\rangle\,.
	\ee
	
	\subsection{Symmetry Breaking Ward Identities for \\ Non-Symmetric Hamiltonians}
	One may consider the case in which the symmetry group $G$ does not commute with time evolution, and the (finite-volume) Hamiltonian $H$ can be written as a $G$-invariant part $H_s$ plus an explicit symmetry breaking term $H_1$, which commutes with $H_s$. To avoid time smearing problems and the possible presence of infinite renormalization constants, we restrict our discussion to the non-relativistic case and assume that the elements of our algebra $\Alg$ admit sharp-time restrictions.
	
	In this case, we assume that the symmetry group $G$ is locally generated at $t=0$ by local charges 
	\be
	Q_R^a = \int d^dx f_R(\mathbf x) j_0^a(\mathbf x, 0)\,,
	\ee
	with $j_0^a$ covariant under space and time translations.
	Then, the correction to the Wigner--Eckart relations, at~$t=0$, is given as before by \eqref{sb_tad}.
	
	The question arises of representing the last term in \eqref{sb_tad} as a tadpole contribution. To this purpose, we~have to generalize \eqref{dadim1}. The proof proceeds as before, with the proviso that the relevant intermediate states contributing to the two-point function \eqref{interm} at $\mathbf k=0$ do not have a gapless energy-momentum dispersion relation.
	
	In order to find such an energy spectrum, since the corresponding (finite-volume) Hamiltonians {$H_s$, $H_1$} commute, the { corresponding} one-parameter groups {$\alpha_s^t$, $\alpha_1^t$} satisfy $\alpha^t = \alpha_s^t \alpha_1^t$. We consider representations in which $\alpha_s^t$ is unbroken, so that the implementability of $\alpha^t$ implies that the ground-state correlation functions are invariant under $\alpha_1^t$.  
	
	Furthermore, we consider the case in which the one-parameter group $\alpha_1^t$ is compact. This condition is satisfied, in particular, whenever $H_1$ is proportional to one of the generators of the compact symmetry group $G$, as is generically the case when the explicit symmetry breaking is due to mass terms or coupling to external fields.
	
	\begin{Proposition}
		Under the above assumptions, let $A$ be an element of $\Alg$ with definite $H_1$ ``charge'': $\alpha_1^t(A)=e^{ih_At} A$. Then, 
		\be
		\lim_{R\to\infty}\langle [\alpha^t(Q_R), A] \rangle = e^{-ih_At}\lim_{R\to\infty}[\langle Q_R, A] \rangle 
		\ee
	\end{Proposition} 
	\begin{proof}
		Since $H_s$ is invariant under the symmetry generated by $Q_R$, one has
		\begin{equation}\begin{aligned}
				\lim_{R\to\infty}\langle [\alpha^t(Q_R), A] \rangle
				=&\lim_{R\to\infty}\langle [\alpha_1^t(Q_R), A] \rangle
				=\lim_{R\to\infty}\langle [(Q_R),\alpha_1^{-t} (A)] \rangle\\
				=& e^{-ih_A t} \lim_{R\to\infty}\langle [Q_R, A] \rangle\,. 
		\end{aligned}\end{equation}
	\end{proof}
	
	This implies that, in the limit $\mathbf k \to 0 $, the spectrum of the intermediate states contributing to the above two-point function
	satisfies 
	\be\label{gap}
	\lim_{\mathbf k\to0} \omega(\mathbf k) = h_A\,,
	\ee
	and, thus, we have the analog of \eqref{dadim1} with $k_0$ satisfying the above dispersion relation, rather than $\omega(\mathbf k)\to0$ as $\mathbf k \to0$.

	\subsection{An Example}
	In order to illustrate the above results,
	we consider the following model, describing non-relativistic, complex scalar fields $\psi_j(t,\mathbf x)$, for $j=1,2$, with (formal) Hamiltonian density
	\be
	\mathcal H = \frac{1}{2m}\nabla \psi^\ast_j \nabla \psi_j+\frac{\lambda}{2} (\psi_j^\ast \psi_j-v^2)^2\,,
	\ee
	where $\lambda$, $v$ are positive constants, and the fields satisfy the standard canonical commutation relations
	$
	[\psi_j(0,\mathbf x), \psi^\ast_k (0, \mathbf y)]=\delta\indices{_{jk}} \delta(\mathbf x - \mathbf y).
	$
	This model exhibits an $SU(2)$ symmetry, whose infinitesimal form reads
	\be
	\delta^a \psi_j = \frac{i}{2} \sigma\indices{^{a}_{jk}}\psi_k = i \lim_{R\to\infty}[Q_R^a, \psi_j]\,,
	\ee
	as in \eqref{limite} with $ j_0^a =-\frac{1}{2}\psi^\ast_j \sigma\indices{^{a}_{jk}}\psi_k$
	($\sigma^a$ are the Pauli matrices), together with a $U(1)$ symmetry \mbox{$\beta^\alpha(\psi_j)= \exp{(i\alpha)}\psi_j$.}
	
	We consider a representation with nonvanishing ground-state expectation of $\psi_1$, \textit{i.e.}, $\langle \psi_1 \rangle = v\neq 0$, leading to spontaneous breaking, in particular, of the subgroup corresponding to $\sigma^1$, since $\langle \delta^1 \psi_2\rangle= iv/2$. We introduced 
	$
	\phi_1\equiv\psi_1-\langle \psi_1 \rangle
	$
	and
	$
	\phi_2\equiv\psi_2\,,
	$
	which parametrize the fluctuations around a minimum of the potential.
	The Goldstone bosons are obtained by applying $\phi_2$ to the ground state.
	
	According to the general strategy described in Section \ref{sec: tadpoles}, it is convenient to consider the Ward identity of Equation \eqref{sb_tad}, with $A_m$, $B_i$ and $C_j$ replaced by $H$, $\psi_1^\ast(\mathbf p)$ and $\psi^\ast_2(\mathbf q)$ respectively. Then, one has
	\begin{equation}\begin{aligned}
		0&=\langle \psi_1(\mathbf p) \delta  H \psi_2^\ast(\mathbf q) \rangle\\
		&=-\langle \delta \psi_1(\mathbf p)H \psi^\ast_2(\mathbf q)\rangle - \langle \psi_1(\mathbf p)H \delta \psi^\ast_2(\mathbf q)\rangle+\langle \delta\large( \psi_1(\mathbf p)  H \psi_2^\ast(\mathbf q) \large)\rangle
	\end{aligned}\end{equation}
	Considering the specific infinitesimal variation $\delta^1$ associated to $\sigma^1$, this gives the ``energy splitting'' between the states obtained by applying $\phi_1^\ast(\mathbf p)$ and $\phi_2^\ast(\mathbf q)$ to the ground state, namely
	\be\begin{aligned}
		&\langle \phi_1(\mathbf p)H \phi^\ast_1(\mathbf q)\rangle - \langle \phi_2(\mathbf p)H \phi^\ast_2(\mathbf q)\rangle
		=
		\langle \psi_1(\mathbf p)H \psi^\ast_1(\mathbf q)\rangle - \langle \psi_2(\mathbf p)H \psi^\ast_2(\mathbf q)\rangle\\
		&=-2i\langle \delta \psi_1(\mathbf p)H \psi^\ast_2(\mathbf q)\rangle -2i \langle \psi_1(\mathbf p)H \delta \psi^\ast_2(\mathbf q)\rangle\,.
	\end{aligned}\ee 
	The analog of \eqref{tadpole_FINALE} then reads
	\be\label{esempio_WI}
	\begin{aligned}
		&\langle \phi_1(\mathbf p)H \phi^\ast_1(\mathbf q)\rangle - \langle \phi_2(\mathbf p)H \phi^\ast_2(\mathbf q)\rangle\\
		&= 2i (2\pi)^{3/2} \lim_{\mathbf k \to 0} F(\mathbf k^2, k_0)
		\big(\langle \mathbf k| \phi_1(\mathbf p) H \phi_2^\ast(\mathbf q)\rangle - \langle \phi_1(\mathbf p) H \phi_2^\ast(\mathbf q) |\mathbf k \rangle \big)\,.
	\end{aligned}
	\ee 
	The right-hand side is easy to calculate in the tree approximation, according to which $\phi_1$ and $\phi_2$ are Fock fields: $F(\mathbf k^2, k_0) = - \frac{i}{2}v$,
	\be\begin{aligned}
		\langle \psi_1(\mathbf p) H \psi_2^\ast(\mathbf q) |\mathbf k \rangle=0\,,\qquad
		\langle \mathbf k| \psi_1(\mathbf p) H \psi_2^\ast(\mathbf q)\rangle=
		\frac{\lambda v}{(2\pi)^{3/2}} \delta(\mathbf q-\mathbf p - \mathbf k)\,,
	\end{aligned}
	\ee
	since the only nonvanishing contributions arise from the cubic term $\lambda v \phi^\ast_1 \phi_2^* \phi_2$ in the potential.
	This gives the ``mass splitting'' $\lambda v^2$ (in agreement with the result of the Fock approximation in the left-hand side). 
	
	On the other hand, the residual unbroken group $\gamma^\alpha(\psi_1)=\psi_1$, $\gamma^\alpha(\psi_2)=e^{i\alpha}\psi_2$ gives rise to Ward identities with no tadpole correction, \textit{i.e.}, to standard Wigner--Eckart relations. In particular, we have selection rules of the type
	$
	\langle B^\ast H C \rangle = 0\,,
	$
	whenever the states obtained applying $B$ and $C$ to the ground state are eigenstates of the unitary operator that implements $\gamma^\alpha$ with different eigenvalues; for instance, $\langle \psi_1 H \psi_2^\ast\rangle=0$.
	
	We may also introduce an explicit symmetry breaking  term
	$
	\mathcal H_1 = \mu \psi^\ast_2 \psi_2
	$
	in the (finite-volume) Hamiltonian density so that the total (finite-volume) Hamiltonian consists of a symmetric term and a symmetry breaking term $H_1$. Since in the infinite-volume limit $\alpha^t=\alpha_s^t \alpha_1^t$, the existence of $H_s$ implies the  implementability $\alpha_1^t$ and, hence, $\langle \psi_2 \rangle =0$. Therefore, as before we consider the case $\langle \psi_1 \rangle =v\neq 0$. In this case, taking into account the explicit breaking $\delta^1 \mathcal H = -\mu j_0^{(2)}$, and that the Goldstone modes acquire an energy gap $\mu$ due to
	\be
	\lim_{R\to\infty}\langle [\alpha^t(Q^1_R),\psi_1(\mathbf p) H \psi_2^\ast(\mathbf q)]\rangle= e^{-it\mu}\lim_{R\to\infty}\langle [Q^1_R,\psi_1(\mathbf p) H \psi_2^\ast(\mathbf q)]\rangle\,,
	\ee
	we can repeat above derivation to obtain the energy splitting at zero momentum $\lambda v^2-\mu$ between the states obtained by applying $\phi_1^\ast(\mathbf p)=\psi_1^\ast(\mathbf p)-v$ and $\phi_2^\ast(\mathbf q)=\psi_2^\ast(\mathbf q)$ to the ground state (in agreement with the \eqref{sb_tad}, \eqref{tadpole_FINALE} and \eqref{gap}).

	\vspace{6pt}
	
\subsection*{Acknowledgements} The research of CH was supported in part by Scuola Normale Superiore, Pisa, by INFN, Sezione di Pisa, and by the Knut
	and Alice Wallenberg Foundation under grant KAW 2018.0116.
	
\providecommand{\href}[2]{#2}\begingroup\raggedright\endgroup

\end{document}